\documentclass[conference]{IEEEtran}
\IEEEoverridecommandlockouts 	
\usepackage{amsmath,amssymb}
\usepackage{amsthm}
\usepackage{color}
\usepackage[mathscr]{eucal}
\usepackage{graphics,graphicx,multicol}
\usepackage{epsfig}
\usepackage{enumerate}
\usepackage{subfigure}
\usepackage{algorithmic}
\usepackage{algorithm}
\usepackage{morefloats}
\usepackage{enumerate}
\usepackage{subfigure}
\usepackage{multirow}
\usepackage{array}
\usepackage{pstricks, pst-node, pst-plot, pst-circ}
\usepackage{moredefs}
\usepackage{cite}
\usepackage{courier}
\newtheorem{thm}{Theorem}[section]
\newtheorem{lem}[thm]{Lemma}

\usepackage{epstopdf}
\DeclareGraphicsExtensions{.eps}
\begin{document}

\title{Towards an Application-Aware Resource Scheduling with Carrier Aggregation in Cellular Systems}
\author{Haya Shajaiah, Ahmed Abdelhadi, and T. Charles Clancy \\
Bradley Department of Electrical and Computer Engineering\\
Hume Center, Virginia Tech, Arlington, VA, 22203, USA\\
\{hayajs, aabdelhadi, tcc\}@vt.edu
}
\maketitle

\begin{abstract}
In this paper, we introduce an application-aware approach for resource block scheduling with carrier aggregation in Long Term Evolution Advanced (LTE-Advanced) cellular networks.
In our approach, users are partitioned in different groups based on the carriers coverage area. In each group of users, users equipments (UE)s are assigned resource blocks (RB)s from all in band carriers. We use a utility proportional fairness (PF) approach in the utility percentage of the application running on the UE. Each user is guaranteed a minimum quality of service (QoS) with a priority criterion that is based on the type of application running on the UE. We prove that our scheduling policy exists and therefore the optimal solution is tractable. Simulation results are provided to compare the performance of the proposed RB scheduling approach with other scheduling policies.
\end{abstract}

\pagenumbering{gobble}
\providelength{\AxesLineWidth}       \setlength{\AxesLineWidth}{0.5pt}%
\providelength{\plotwidth}           \setlength{\plotwidth}{8cm}
\providelength{\LineWidth}           \setlength{\LineWidth}{0.7pt}%
\providelength{\MarkerSize}          \setlength{\MarkerSize}{3pt}%
\newrgbcolor{GridColor}{0.8 0.8 0.8}%
\newrgbcolor{GridColor2}{0.5 0.5 0.5}%
\vspace{-0.5em}
\section{Introduction}\label{sec:intro}
\vspace{-0.3em}
In recent years, there has been tremendous demands of wireless services with high bandwidth. However, it is difficult to provide the required resources with a single frequency band. To address this issue, carrier aggregation (CA) technique was introduced under LTE-Advanced networks \cite{work-item}. 
Utility PF resource allocation for a single carrier in cellular networks have been extensively studied in \cite{Ahmed_Utility1}. The problem of RA for multi-carrier systems in single cell have been given attention in recent years \cite{Dual-Decomposition, Resource_allocation, Rate_Balancing,Haya_Utility1}. The authors in \cite{Haya_Utility1} have presented multi-stage resource allocation (RA) with CA algorithms for multi-carrier cellular systems. However, non of their RA approaches have considered the problem of RB scheduling for multiple component carriers (CC)s.

In this paper, we focus on solving the problem of utility PF resource block scheduling with CA for multi-carrier cellular networks. The resource scheduling approach presented in \cite{SelfOrganizedLTE,Tugba-RB} does not consider the case of multi-carrier resources available at the eNodeB. It only solves the problem of RB scheduling in the case of single carrier. In this paper, we introduce a user grouping method that creates a user group for each CC such that each carrier assigns its RBs only to users in its user group. Each user is assigned on multiple CCs' RBs based on a user grouping method and a utility PF policy. We prove that the proposed resource scheduling policy, that is based on CA, exists and that the optimal solution is tractable.

%
The remainder of this paper is organized as follows. In Section \ref{sec:Problem_formulation}, we describe the system model and problem setup, the utility functions of users rates and the user grouping method. Section \ref{sec:RBscheduling} presents a RB scheduling with CA optimization problem. In section \ref{sec:sim}, we present our simulation results for the proposed RB scheduling approach and compare its performance with other RB scheduling approaches. Section \ref{sec:conclude} concludes the paper.
\vspace{-0.5em}
\section{System Model and Problem Setup}\label{sec:Problem_formulation}
The transmission resources in a LTE downlink have dimensions in frequency, time and space \cite{LTEBook}. The frequency is represented by subcarriers. The time is divided into frames and each frame is further divided into subframes. The space is provided by the transmit and receive antennas. One RB consists of 12 continuous subcarriers. In reuse-1 radio systems, that is considered in this paper, a RB can be allocated to only one user.

In this paper, we consider a single cell LTE-Advanced mobile system with one evolved NodeB (eNodeB) and $M$ users. Let the number of CCs that the system can aggregate be $K$. The set of CCs is given by $\mathcal{K}=\{f_1,f_2,...,f_K\}$  with CCs in order from the highest frequency to the lowest frequency (i.e. $f_1> f_2>...>f_K$). We consider an equal power allocation (EPA) scheme that each frequency component has the same transmitting power. Furthermore, a non adjacent inter band aggregation scenario is considered. Because the channel fading for high frequency is larger than that for low frequency, higher frequency carriers have smaller coverage areas than lower frequency carriers. Users located under the coverage area of multiple carriers are scheduled resources from all in band carriers. The eNodeB assigns RBs from multiple carriers to each UE. The total allocated rate achieved by assigning RBs to the $i^{th}$ UE is given by $r_i$. Each UE has its own utility function $U_i(r_i)$  that corresponds to the type of application running on the $i^{th}$  UE. Our goal is to determine which RBs from each CC should be allocated to each UE by the eNodeB in order to maximize the total system utility while ensuring PF between utilities.

We define $\mathcal{Z}_k$, where $1\leq k \leq K$, to be the set of RBs available by $f_k$ carrier where $z_{k,j}$ denotes a single RB in $\mathcal{Z}_k=\{z_{k,1},z_{k,2},...\}$, $z_{k,j} \in \mathcal{Z}_k$ is the $j^{th}$ RB in CC $f_k$ and $|\mathcal{Z}_k|$ denotes the number of RBs available by $f_k$ carrier. The signal to noise ratio (SNR) of user $i$ on RB $z_{k,j}$ is given by $\gamma_{i,z_{k,j}} = P_{z_{k,j}} |G_{i,z_{k,j}}|^2 / N_{i,z_{k,j}}$
where $G_{i,z_{k,j}}$ is the complex channel gain between the eNodeB and the $i^{th}$ UE on RB $z_{k,j}$, $N_{i,z_{k,j}}$ is the noise power experienced by the $i^{th}$ UE on RB $z_{k,j}$ and $P_{z_{k,j}}$ is the transmission power that the eNodeB assigns to RB $z_{k,j}$. Under the EPA, $P_{z_{k,j}}=P_k/|\mathcal{Z}_k|$ where $P_k$ is the transmitting power of CC $f_k$. Then the achievable data rate of the $i^{th}$ user on RB $z_{k,j}$ is given by
\vspace{-0.7em}
\begin{equation}\label{eqn:datrate on z_{k,j}}
H_{i,z_{k,j}} = W \log(1+\beta_{z_{k,j}} \gamma_{i,z_k}),
\end{equation}
where $W$ is the bandwidth of a RB and $\beta_{z_{k,j}}$ is the SNR gap.

In each frame, the eNodeB schedules each of the frame's RBs to one UE. Let $\phi_{i,z_{k,j}}$ be the proportion of frames that the $i^{th}$ UE is scheduled by the eNodeB on RB $z_{k,j}$. The $i^{th}$ UE rate on all RBs scheduled by carrier $f_k$ is given by
\begin{equation}\label{eqn:rate by f_k}
r_{i,f_k} = \sum_{z_{k,j}\in \mathcal{Z}_k}\phi_{i,z_{k,j}} H_{i,z_{k,j}}.
\end{equation}

The overall rate of the $i^{th}$ UE, that is the sum of the rates achieved by all carriers RBs assignments, is given by $r_i = \sum_{f_k\in \mathcal{K}}r_{i,f_k}$.

A user grouping method is introduced in \ref{sec:UsersGrouping} to partition users into groups depending on their location in the cell. The eNodeB performs RBs assignments from each CC to the user group located in the coverage area of that carrier.
\vspace{-0.5em}
\subsection{Utility Functions of Users Rates}\label{sec:utility functions}
We express the user satisfaction with its application rates using utility functions. We represent the $i^{th}$ user application utility function $U_i (r_i)$  by  sigmoidal-like function or logarithmic function where $r_i$ is the rate of the $i^{th}$ user application. These utility functions have the following properties: 1) $U_i(0) = 0$ and $U_i(r_i)$ is an increasing function of $r_i$. 2) $U_i(r_i)$ is twice continuously differentiable in $r_i$ and bounded above.

In our model, we use the normalized sigmoidal-like utility function to represent real-time applications, same as the one presented in \cite{Ahmed_Utility1}, that is
\begin{equation}\label{eqn:sigmoid}
U_i(r_i) = c_i\Big(\frac{1}{1+e^{-a_i(r_i-b_i)}}-d_i\Big),
\end{equation}
where $c_i = \frac{1+e^{a_ib_i}}{e^{a_ib_i}}$ and $d_i = \frac{1}{1+e^{a_ib_i}}$ so it satisfies $U_i(0)=0$ and $U_i(\infty)=1$. The normalized sigmoidal-like function has an inflection point at $r_i^{\text{inf}}=b_i$. In addition, we use the normalized logarithmic utility function to represent delay-tolerant applications, same as the one used in \cite{Ahmed_Utility1}, that can be expressed as
\begin{equation}\label{eqn:log}
U_i(r_i) = \frac{\log(1+k_ir_i)}{\log(1+k_ir_{\text{max}})},
\end{equation}
where $r_{\text{max}}$ gives $100\%$ utilization and $k_i$ is the slope of the curve that varies based on the user application. So, it satisfies $U_i(0)=0$ and $U_i(r_{\text{max}})=1$.
\vspace{-0.5em}
\subsection{User Grouping Method}\label{sec:UsersGrouping}
In this section we introduce a user grouping method to create one user group $\mathcal{M}_{f_k}$ for each CC $f_k$ where $\mathcal{M}_{f_k}$ is a set of users located under the coverage area of carrier $f_k$. Users in $\mathcal{M}_{f_k}$ are assigned RBs on CC $f_k$ by the eNodeB. Users located under the coverage area of multiple carriers (i.e. common users in multiple user groups) are assigned RBs on these carriers and their final rates are aggregated under a non adjacent inter band aggregation scenario.

The $i^{th}$ user is part of user group $\mathcal{M}_{f_k}$ if it satisfies certain path loss constraints on CC $f_k$. Assume that the maximum pathloss in a carrier can not exceed a threshold $L^{th}$. In order for the eNodeB to identify a user group for each CC, it first computes the $i^{th}$ user pathloss on each CC and creates a set $\alpha_i$ that includes all in range carriers such that the $i^{th}$ user is assigned RBs only from carriers in $\alpha_i$.
%

Higher frequency carriers have smaller coverage radius $R_k$ than lower frequency carriers (i.e. $R_1<R_2<...<R_K$). Therefore, user group $\mathcal{M}_{f_1} \subseteq \mathcal{M}_{f_2} \subseteq ...\subseteq \mathcal{M}_{f_K}$. 
\section{RB Scheduling with CA Problem}\label{sec:RBscheduling}
In this section, we present our RB scheduling with CA approach. Our objective is to assign RBs to each user (i.e. the $i^{th}$ user) on all of its in range carriers (i.e. CCs in $\alpha_i$) based on a utility PF policy. We use utility functions of users' applications rates to represent the type of application running on the UE. Given that different applications may have different QoS requirements, every user subscribing for a mobile service is guaranteed to achieve minimum QoS for each of its applications with a priority criterion. Users running real-time applications are given priority when assigning RBs due to the sigmoidal-like utility functions nature used to represent their applications. In addition, our utility PF approach guarantees that no user is assigned zero RBs.

The eNodeB performs the RBs assignment for each of the CC's RBs in $\mathcal{Z}_k$. It assigns the RBs of each CC $f_k$ one at a time and one after another in ascending order of their coverage radius $R_k$. It starts with CC $f_1$ as it has the smallest coverage radius $R_1$. After assigning all users in $\mathcal{M}_{f_1}$ on $f_1$ RBs, the eNodeB then assigns users in $\mathcal{M}_{f_2}$ on $f_2$ RBs. In addition, since $\mathcal{M}_{f_1}$ users are also in $\mathcal{M}_{f_2}$ (i.e. $\mathcal{M}_{f_1} \subseteq \mathcal{M}_{f_2}$), the eNodeB assigns $\mathcal{M}_{f_1}$ users on $f_2$ RBs and the rates are aggregated based on a non adjacent inter band aggregation scenario. The eNodeB continues the RB assignment process by assigning $\mathcal{M}_{f_k}$ users on CC $f_k$ RBs. Finally, the RB assignment process is finalized by assigning carrier $f_K$ RBs to all users in the cellular cell as they are all located within its coverage radius. We consider a utility PF objective function, based on CA, that the eNodeB seeks to maximize for each time it assigns user on a RB. The utility PF resource scheduling with CA optimization problem for the eNodeB assignments of $\mathcal{M}_{f_k}$ users on $\mathcal{Z}_k$ RBs is given by
\begin{equation}\label{eqn:utility_fairness_CA}
\begin{aligned}
& \underset{\textbf{$\phi_{i,z_k}$}} {\text{max}}
& & \prod_{i=1}^{M_k} U_i\Big (c_{i,f_k}+\sum_{z_{k,j}\in \mathcal{Z}_k}(\phi_{i,z_{k,j}} H_{i,z_{k,j}})\Big ) \\
& \text{subject to}
& & \sum_{i=1}^{M_k}\phi_{i,z_{k,j}} =1, \;\;c_{i,f_1} = 0,\\
& & &  \phi_{i,z_{k,j}} \geq 0, \;\;\; i = 1,2, ...,M_k\\
& & &  c_{i,f_k} = \sum_{l=1}^{k-1}r_{i,f_l}, \;\; k>1.
\end{aligned}
\end{equation}
where $M_k=|\mathcal{M}_{f_k}|$ is the number of UEs in the coverage area of carrier $f_k$,
$c_{i,f_1}=0$ and $c_{i,f_k}$ for $k>1$ is equivalent to $\sum_{l=1}^{k-1}r_{i,f_l}$ that is the $i^{th}$ UE total rate on all RBs scheduled by carriers $\{f_1,...,f_{k-1}\}$. The eNodeB seeks to maximize the objective function of this resource scheduling optimization problem that is achieved by maximising the product of all UEs' utilities when assigning the UEs on the carriers' RBs. The goal of this resource scheduling objective function is to allocate the resources to the UE that maximizes the total cellular network objective (i.e. the product of the utilities of all UEs) while ensuring PF between individual utilities. This objective function ensures non-zero RA for all users. Therefore, the resource scheduling optimization problem guarantees minimum QoS for all users. In addition, this approach allocates more resources to real-time applications providing improvement to the QoS of LTE system.

Later in this section we prove that there exists a tractable global optimal solution to optimization problem (\ref{eqn:utility_fairness_CA}). However, the user's final rate, achieved by assigning each user on its in range carriers' RBs, is determined using a multi-stage approach where optimization problem (\ref{eqn:utility_fairness_CA}) is required for each CC $f_k$. In addition, optimization problem (\ref{eqn:utility_fairness_CA}) needs to be applied in a multi-stage scenario starting from the carrier with the smallest coverage area (i.e. $f_1$) and ending with the carrier that has the largest coverage area (i.e. $f_K$). The rate achieved for each user after assigning CC $f_k$ RBs is needed for the next stage optimization problem (\ref{eqn:utility_fairness_CA}) of carrier $f_{k+1}$. The objective function in optimization problem (\ref{eqn:utility_fairness_CA}) is equivalent to $\arg \underset{\textbf{\textbf{$\phi_{i,z_k}$}}} \max \sum_{i=1}^{M_k}\log (U_i(c_{i,f_k}+\sum_{z_{k,j}\in \mathcal{Z}_k}(\phi_{i,z_{k,j}} H_{i,z_{k,j}})))$. The utility functions $\log(U_i(c_{i,f_k}+\sum_{z\epsilon Z}\phi_{i,b(i),z} H_{i,b(i),z}))$ that are equivalent to $\log(U_i(c_{i,f_k}+r_{i,f_k})$ are strictly concave functions as proved in \cite{Ahmed_Utility1}. As a result, optimization problem (\ref{eqn:utility_fairness_CA}) is a convex optimization problem and there exists a unique tractable global optimal solution \cite{Ahmed_Utility1,Haya_Utility1}.
%

In order to consider the case when the entire input is not available from the beginning, we use an online algorithm as in \cite{SelfOrganizedLTE,Tugba-RB}. The total achieved data rate of each UE when assigning it on different CCs' RBs, i.e. $r_i$, requires the knowledge of $\phi_{i,z_{k,j}}$ on each RB $z_{k,j}$ the UE is assigned on. We use an online scheduling algorithm to decrease the computation overhead while processing the rate information as in \cite{SelfOrganizedLTE}.

Let $\phi_{i,z_{k,j}}[n]$ be the proportion of the frames that UE $i$ is scheduled on RB $z_{k,j}$ in the first $n$ frames. Then, the proportion of the frames that UE $i$ is scheduled on RB $z_{k,j}$ in the $[n+1]^{th}$ frame is defined as follows:
\vspace{-0.5em}
\begin{equation}\label{eqn:online_algorithm}
\begin{aligned}
\phi_{i,z_{k,j}}[n+1]=
\begin{cases}
	\frac{n-1}{n}\phi_{i,z_{k,j}}[n]+\frac{1}{n},\\
	\text{if UE $i$ is scheduled on RB $z_{k,j}$}\\
	\text{in the $(n+1)^{th}$ frame}\\
	\frac{n-1}{n}\phi_{i,z_{k,j}}[n],\text{otherwise}.
\end{cases}
\end{aligned}
\end{equation}

In the proposed scheduling policy, for certain CC's RB $z_{k,j}$, the eNodeB schedules the UE that maximizes $\frac{U'_i(c_{i,f_k}+\sum_{z_{k,j}\in \mathcal{Z}_k}\phi_{i,z_{k,j}}H_{i,z_{k,j}})H_{i,z_{k,j}}}{U_i(c_{i,f_k}+r_{i,f_k})}$ on RB $z_{k,j}$.

\begin{lem}\label{lem:optimality}  
Using the scheduling policy in (\ref{eqn:online_algorithm}), we show that $\lim \inf_{n \rightarrow \infty} \sum_{i=1}^{M_k} \log U_i(c_{i,f_k}+\sum_{z_{k,j}\in \mathcal{Z}_k}(\phi_{i,z_{k,j}}[n] H_{i,z_{k,j}}))$ exists for optimization problem (\ref{eqn:utility_fairness_CA}).
\end{lem}

\begin{proof}
We define $L(\phi)=\sum_{i=1}^{M_k}\log U_i(c_{i,f_k}+\sum_{z_{k,j}\in \mathcal{Z}_k}(\phi_{i,z_{k,j}} H_{i,z_{k,j}}))$ where $\phi$, $\phi[n]$ and $H$ are the short terms for $\phi_{i,z_{k,j}}$, $\phi_{i,z_{k,j}}[n]$ and $H_{i,z_{k,j}}$, respectively. Let $r_{i,f_k}[n]=\sum_{z_{k,j}}(\phi_{i,z_{k,j}}[n] H_{i,z_{k,j}})$. Using Taylor's theorem, for any $\phi$ and $\Delta\phi$ we have

$L(\phi+\Delta\phi)=L(\phi)+L'(\phi)\Delta\phi+\pi(\phi,\Delta\phi)$

where $|\pi(\phi+\Delta\phi)|<b |\Delta\phi|^2$, for some constant $b$.

Let $\Delta\phi_{i,z_{k,j}}[n]=\phi_{i,z_{k,j}}[n+1]-\phi_{i,z_{k,j}}[n]$, then
\begin{equation*}\label{eqn:delta_phi_eqn}
\begin{aligned}
\Delta\phi_{i,z_{k,j}}[n]=
\begin{cases}
	\frac{1}{n}-\frac{\phi_{i,z_{k,j}}[n]}{n},\\
	\text{if UE $i$ is scheduled on RB $z_{k,j}$}\\
	\text{in the $(n+1)^{th}$ frame}\\
	\frac{-\phi_{i,z_{k,j}}[n]}{n},\text{otherwise}.
\end{cases}
\end{aligned}
\end{equation*}

$|\Delta\phi_{i,z_{k,j}}[n]|<\frac{1}{n}$, for all $i$ and $z_{k,j}$. As a result;

\vspace{-1em}
\begin{equation}\label{eqn:L_inequality}
\begin{aligned}
L(\phi[n+1]) & = L(\phi[n]+\Delta\phi[n]),\\
& \geq L(\phi[n])+\Delta
L(\phi[n])-\frac{b}{n^2},\\
&=L(\phi[n])
+\Big(\sum_i \frac{U'_i(c_{i,f_k}+\sum_{z_{k,j}}\phi H)}{U_i(c_{i,f_k}+r_{i,f_k})}\\
&H \Delta\phi \Big) -\frac{b}{n^2} \\
&=L(\phi[n])
+\frac{1}{n}\Big(\max_i\frac{U'_i(c_{i,f_k}+\sum_{z_{k,j}}\phi H)}{U_i(c_{i,f_k}+r_{i,f_k})}\\
&H -\sum_i\frac{U'_i(c_{i,f_k}+\sum_{z_{k,j}}\phi H)}{U_i(c_{i,f_k}+r_{i,f_k})}
H \phi[n]\Big)-\frac{b}{n^2}\\
&\geq L(\phi[n])-\frac{b}{n^2},
\end{aligned}
\end{equation}
where $\Delta\phi[n]$ is substituted by $(\frac{1}{n}-\frac{\phi_{i,z_{k,j}}[n]}{n})$ (i.e. user $i$ has the largest $\frac{U'_i(c_{i,f_k}+\sum_{z_{k,j}}\phi H)H}{U_i(c_{i,f_k}+r_{i,f_k})}$ among all users) and the last inequality holds since $\sum_i\phi_{i,z_{k,j}[n]}=1$ for all $i$ and $z_{k,j}$.

Let $\beta := \lim \sup_{n \rightarrow \infty} L(\phi[n])$. For any $\epsilon > 0$, there exists large enough $N$ so that $L(\phi[N])>\beta-\frac{\epsilon}{2}$ and $\sum_{n=N}^{\infty}\frac{b}{n^2}<\frac{\epsilon}{2}$. For any $\hat{n}>N$, $L(\phi[\hat{n}]) \geq L(\phi[N])-\sum_{n=N}^{\hat{n}}\frac{b}{n^2} > \beta-\epsilon$. Therefore, $L(\phi[n])$ converges to $\beta$, as $n \rightarrow \infty$.

Due to the constraint $\sum_{i=1}^{M_k} \phi_{i,z_{k,j}} = 1$ in (\ref{eqn:utility_fairness_CA}), $\phi$ is a solution to optimization problem (\ref{eqn:utility_fairness_CA}) if and only if
\vspace{-0.5em}
\begin{equation}\label{eqn:opt_Maximization}
\begin{aligned}
\frac{dL}{d\phi_{i,z_{k,j}}}  &=\frac{U'_i(c_{i,f_k}+\sum_{z_{k,j}}\phi H)H}{U_i(c_{i,f_k}+r_{i,f_k})} \\
&=\max_m \frac{U'_m(c_{m,f_k}+\sum_{z_{k,j}}\phi_{m,z_{k,j}}H_{m,z_{k,j}})}{U_m(c_{m,f_k}+r_{m,f_k})}
H_{m,z_{k,j}},
\end{aligned}
\end{equation}
for all $i$ and $z_{k,j}$ such that $\sum_{i=1}^{M_k}\phi_{i,z_{k,j}}=1$ and $\phi_{i,z_{k,j}} \geq 0$.
\end{proof}
\vspace{-1.3em}
\begin{thm}\label{thm:limitmax}
Using the scheduling policy (\ref{eqn:opt_Maximization}), $\lim_{n \rightarrow \infty}L(\phi)[n]=\sum_{i=1}^{M_k}\log U_i(c_{i,f_k}+\sum_{z_{k,j}}(\phi_{i,z_{k,j}}[n] H_{i,z_{k,j}}))$ (i.e. $\lim_{n \rightarrow \infty}L(\phi[n])$) achieves the maximum of optimization problem (\ref{eqn:utility_fairness_CA}).
\end{thm}
\begin{proof}
Suppose $\lim_{n \rightarrow \infty}L(\phi[n])$ does not achieve the maximum of the optimization problem. There exists $\delta >   0$, $\lambda > 0$, and positive integer $N$ such that for all $n > N$, there exists some $i^n \in M_k$ and $z_{k,j}^n \in \mathcal{Z}_k$ so that $\phi_{i^n,z_{k,j}^n}[n]>\delta$ and  $\frac{U'_{i^n}(c_{i^n,f_k}+\sum_{z_{k,j}}\phi_{i^n,z_{k,j}^n}H_{i^n,z_{k,j}^n})H_{i^n,z_{k,j}^n}}{U_i(c_{i^n,f_k}+r_{i^n,f_k})}<\max_m \frac{U'_m(c_{m,f_k}+\sum_{z_{k,j}}\phi_{m,z_{k,j}^n}H_{m,z_{k,j}^n})H_{m,z_{k,j}^n}}{U_m(c_{m,f_k}+r_{m,f_k})}-\lambda$. Now we have:
\vspace{-0.7em}
\begin{equation*}\label{eqn:L_maximize}
\begin{aligned}
& L(\phi[n+1])-L(\phi[n]) \geq L'(\phi[n]))\Delta \phi[n]-\frac{b}{n^2}\\
& =\sum_{i=1}^{M_k}\frac{U'_i(c_{i,f_k}+\sum_{z_{k,j}}\phi[n]H)H}{U_i(c_{i,f_k}+r_{i,f_k})}
\Delta\phi[n]-\frac{b}{n^2}\\
& =\frac{\delta\lambda}{n}-\frac{b}{n^2} \geq \frac{\delta\lambda}{2n},
\vspace{-1em}
\end{aligned}
\end{equation*}
for large enough $n$. Since $\sum_{n=1}^\infty \frac{1}{n}=\infty$, which is a contradiction. As a result, $\lim_{n \rightarrow \infty}L(\phi[n])$ achieves the maximum of the optimization problem.
\end{proof}
\vspace{-0.08in}
\section{Simulation Results}\label{sec:sim}
In this section we present simulation results for the proposed resource scheduling with CA approach. We consider a LTE-Advanced mobile system with $M=8$ users and two CCs $f_1$ and $f_2$ available at the eNodeB with $f_1>f_2$ as shown in Figure \ref{fig:SystemModel2}. We apply the user grouping method presented in \ref{sec:UsersGrouping} and two user groups are obtained, $\mathcal{M}_{f_1}=\{1,2,3,4\}$ and $\mathcal{M}_{f_2}=\{1,2,...,8\}$ where user $i \in \mathcal{M}_{f_k}$ represents the $i^{th}$ user located under the coverage area of carrier $f_k$. Users $\{1,2,5,6\}$ are running real-time applications that are represented by sigmoidal-like utility functions with parameters $a_i = 5$ and $b_i=10$ for users $\{1,5\}$ and $a_i = 1$ and $b_i=30$ for users $\{2,6\}$. Users $\{3,4,7,8\}$ are running delay-tolerant applications that are represented by logarithmic utility functions with parameters $k_i = 15$ for users $\{3,7\}$ and $k_i = 0.5$ for users $\{4,8\}$. The simulation was run using MATLAB.
%

We compare the performance of the resource scheduling with CA approach in the case of using the proposed utility proportional fairness (UPF) resource scheduling policy and in the case of using the traditional proportional fairness (traditional-PF) scheduling policy presented in \cite{SelfOrganizedLTE}. We assume equal channel gain in our simulation. In Figure \ref{fig:PerformanceC1C2}, we show simulation results and compare the performance of different scheduling policies for users in $\mathcal{M}_{f_1}$ that are assigned RBs by carrier $f_1$ and users in $\mathcal{M}_{f_2}$ that are assigned RBs by carrier $f_1$ and $f_2$. Figure \ref{fig:PerformanceC1C2} shows the objective function of carrier $f_1$ RA optimization problem that is given by the multiplication of all users' applications quality of experience (QoE) for users in $\mathcal{M}_{f_1}$ and the objective function of carrier $f_2$ RA optimization problem when using the aforementioned scheduling policies. 
Figure \ref{fig:PerformanceC1C2} shows that the system performance, represented by the objective function value of the RA optimization problem that is given by the multiplication of all users applications' utilities, that represent users' satisfaction with the allocated rates in the case of the proposed UPF scheduling policy is much greater than the objective function value when using the traditional-PF scheduling policy. It also shows that the system performance when using the traditional-PF with equal priority weights is worse than the system performance when using the traditional-PF with non equal priority weights.
\begin{figure}
\centering
\vspace{-1em}
\includegraphics[height=2.2in, width=2.3in]{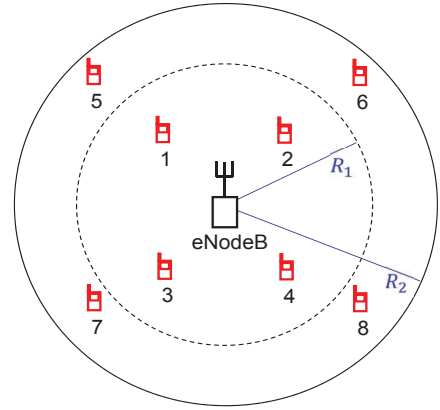}
\vspace{-1em}
\caption{LTE-Advanced mobile system with two component carriers (i.e. $f_1$ and $f_2$) available at the eNodeB with $f_1>f_2$ and $R_1<R_2$.}
\vspace{-0.2em}
\label{fig:SystemModel2}
\end{figure}
\begin{figure}
\centering
\vspace{-1em}
\includegraphics[height=1.85in, width=3.5in]{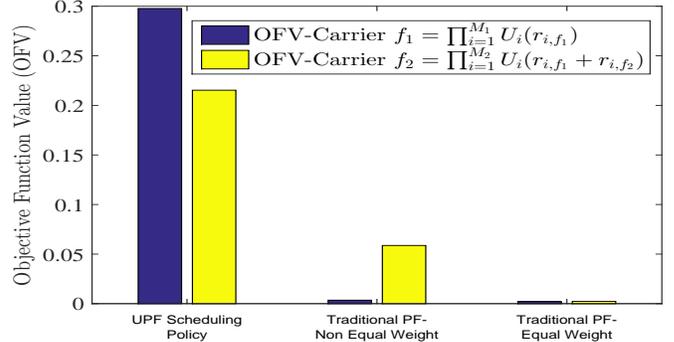}
\vspace{-1.5em}
\caption{Performance comparison for different scheduling policies represented by the objective function of carrier $f_1$ and $f_2$ RA optimization problems.}
\vspace{-1.2em}
\label{fig:PerformanceC1C2}
\end{figure}
%
%
%
%
%
\section{Conclusion}\label{sec:conclude}
In this paper, we introduced a RB scheduling with CA approach in LTE-Advanced. Users are partitioned in user groups and each user is assigned on RBs of its corresponding in range carriers. We used utility PF with CA policy and presented users applications using utility functions. We proved that our scheduling policy exists and therefore the optimal solution is tractable. Simulation results showed that the proposed resource scheduling with CA policy achieves better QoE than the traditional proportional fairness policy. 

\bibliographystyle{ieeetr}
\bibliography{pubs}
\end{document}